\documentclass{article}

\usepackage{fullpage, amsmath, amsthm, amssymb}
\usepackage[noend,algoruled,linesnumbered,]{algorithm2e}

\DeclareMathOperator{\disc}{disc}
\DeclareMathOperator{\herdisc}{herdisc}

\DeclareMathOperator{\tr}{tr}
\DeclareMathOperator{\rank}{rank}
\newcommand{\R}{\mathbb{R}}
\newcommand{\Norm}{\mathcal{N}}
\newcommand{\E}{\mathbb{E}}
\newcommand{\Uni}{\mathcal{U}}
\newcommand{\Fam}{\mathcal{S}}
\newcommand{\eps}{\varepsilon}
\newcommand{\Minimize}{\textsc{HereditaryMinimize}\xspace}
\newcommand{\Sample}{\textsc{Sample}\xspace}
\newcommand{\SampleMany}{\textsc{SampleMany}\xspace}

\begin{document}

\newtheorem{theorem}{Theorem}
\newtheorem{lemma}{Lemma}
\newtheorem{claim}{Claim}

\title{Fast Discrepancy Minimization with Hereditary Guarantees}

\author{
  Kasper Green Larsen\thanks{Supported by Independent Research Fund Denmark (DFF) Sapere Aude Research Leader grant
No 9064-00068B.}\\
  Aarhus University\\
  \texttt{larsen@cs.au.dk}
}

\date{}

\maketitle

\begin{abstract}
Efficiently computing low discrepancy colorings of various set systems, has been studied extensively since the breakthrough work by Bansal (FOCS 2010), who gave the first polynomial time algorithms for several important settings, including for general set systems, sparse set systems and for set systems with bounded hereditary discrepancy. The hereditary discrepancy of a set system, is the maximum discrepancy over all set systems obtainable by deleting a subset of the ground elements. While being polynomial time, Bansal's algorithms were not practical, with e.g. his algorithm for the hereditary setup running in time $\Omega(m n^{4.5})$ for set systems with $m$ sets over a ground set of $n$ elements. More efficient algorithms have since then been developed for general and sparse set systems, however, for the hereditary case, Bansal's algorithm remains state-of-the-art. In this work, we give a significantly faster algorithm with hereditary guarantees, running in $O(mn^2\lg(2 + m/n) + n^3)$ time. Our algorithm is based on new structural insights into set systems with bounded hereditary discrepancy. We also implement our algorithm and show experimentally that it computes colorings that are significantly better than random and finishes in a reasonable amount of time, even on set systems with thousands of sets over a ground set of thousands of elements.
\end{abstract}  

\thispagestyle{empty}
\newpage
\setcounter{page}{1}

\section{Introduction}
In discrepancy minimization, we are given a set system $\Fam = \{S_1,\dots,S_m\}$, with each set $S_i$ being a subset of a universe $\Uni = \{u_1,\dots,u_n\}$ of $n$ elements. The goal is to find a red-blue coloring of the elements of $\Uni$ such that each set $S_i$ is colored as evenly as possible. More formally, if we define the $m \times n$ incidence matrix $A$ such that entry $(i,j)$ is $1$ if $u_j \in S_i$ and $0$ otherwise, then we seek a ``coloring'' $x \in \{-1,1\}^n$ for which the discrepancy  $\disc(A,x) := \|Ax\|_\infty = \max_i |(Ax)_i|$ is as small as possible. Discrepancy minimization has been studied for decades and has found numerous applications in theoretical computer science, see e.g. the textbooks by Matousek~\cite{matousek1999geometric} and Chazelle~\cite{chazelle:discrepancy}.

A natural generalization of discrepancy minimization is to allow arbitrary real matrices $A \in \R^{m \times n}$, rather than only incidence matrices corresponding to set systems. For a matrix $A \in \R^{m \times n}$, its discrepancy is defined as $\disc(A) = \min_{x \in \{-1,1\}^n} \disc(A,x)$, i.e. as the best possible coloring achievable. Understanding the discrepancy of various matrices has been the focus of much research. The three main directions pursued in the area can be roughly categorized as follows:
\begin{enumerate}
\item Set systems corresponding to geometric ranges. Here the matrices $A$ of interest correspond to incidence matrices, where the elements $\Uni = \{u_1,\dots,u_n\}$ typically are points in $\R^d$ and each set $S_i$ corresponds to the subset of $\Uni$ that lies inside a geometric range such as a half-space or an axis-aligned rectangle~\cite{alexander:half,matousek:half,matousekGamma,larsenDisc}.

\item General set systems, where one only assumes that $|a_{i,j}| \leq 1$ for all $(i,j)$. The classic ``Six Standard Deviations Suffice'' result by Spencer~\cite{spencer} is one of the pillars of this line of research, showing that for any $n \times n$ matrix $A$ with $|a_{i,j}| \leq 1$, we have $\disc(A) \leq 6 \sqrt{n}$. This is highly surprising as a random coloring $x$ for most matrices $A$ would result in a discrepancy of $\disc(A,x) = \Omega(\sqrt{n \lg n})$ with high probability.

\item Sparse set systems, where each element of the universe $u_j \in \Uni$ is contained in at most $t$ sets $S_i$. The Beck-Fiala conjecture~\cite{beckfiala}, which is a special case of the Koml\'{o}s conjecture, is central in this direction, asserting that any such set system can be colored to achieve a discrepancy of only $O(\sqrt{t})$ (the Koml\'{o}s conjecture states that any sequence of vectors with unit $\ell_2$ norm can be assigned signs such that their sum has $\ell_\infty$ norm bounded by a constan). While not yet being resolved, the best upper bound, due to Banaszczyk~\cite{Banas:gauss} comes close by guaranteeing the existence of a coloring of discrepancy $O(\sqrt{t \lg n})$.
\end{enumerate}

\paragraph{Constructive Discrepancy Minimization.}
Much of the early work on discrepancy minimization focused only on showing the existence of low-discrepancy colorings~\cite{spencer,Banas:gauss,srinivasan} and it was initially unclear whether there exist polynomial time algorithms for computing near-optimal colorings. Indeed, a hardness result by Charikar et al.~\cite{Charikar} showed that it is NP-hard to distinguish whether $\disc(A)$ is $0$ or $\Omega(\sqrt{n})$ for matrices $A$ corresponding to set systems. This rules out any polynomial time multiplicative approximation algorithm for discrepancy minimization. While this might seem like an insurmountable obstacle for interesting algorithmic results, Bansal's~\cite{bansal} seminal paper presented a sequence of polynomial time algorithms with different exciting approximation guarantees. With respect to the main directions mentioned above, Bansal gave an algorithm for general set systems guaranteeing to find a coloring of discrepancy $O(\sqrt{n})$ for $n \times n$ set systems. In some sense, this can be thought of as an \emph{additive} $O(\sqrt{n})$ approximation. For the sparse set system case, he gave an algorithm guaranteeing a discrepancy of $O(\sqrt{t} \lg n)$, getting close, but not quite matching the non-constructive bound by Banaszczyk~\cite{Banas:gauss}. Later work~\cite{bansalKomlos} has since then closed the gap and given a polynomial time algorithm matching Banaszczyk's bound. Bansal's final algorithmic result gives a coloring with discrepancy related to the so-called \emph{hereditary discrepancy} of $A$~\cite{lovasz}. The hereditary discrepancy of $A$, denoted $\herdisc(A)$, is a classic measure defined as the maximum discrepancy of any submatrix $A'$ obtained by deleting a subset of the columns of $A$ (equivalently, deleting a subset of the points in $\Uni$). With this definition, his algorithm finds a coloring $x$ such that $\disc(A,x) = O(\lg n \cdot \herdisc(A))$. So while Charikar et al.'s NP-hardness result shows that it is generally hard to find colorings of discrepancy $o(\sqrt{n})$, even when such colorings exist, Bansal's result shows that this is in fact possible if all sub-matrices of $A$ have low-discrepancy colorings. All of Bansal's algorithms are based on semi-definite programming (SDP) and are not as such practically efficient. Concretely, his algorithm with hereditary guarantees needs to solve $\Omega(n^2)$ semi-definite programs, all consisting of $m$ constraint matrices of size $n \times n$. With the current state-of-the-art SDP solvers~\cite{sdp}, this requires time $\Omega(m n^{4.5})$, making it practically infeasible even for modest size matrices. Since Bansal's celebrated work, follow-up work by Lovett and Meka~\cite{lovett} gave an algorithm for the general set system case, running in time $O((m+n)n^2)$ and Alweiss et al.~\cite{fastdisc} recently gave an even faster algorithm for sparse set systems, guaranteeing $\disc(A,x) = O(\sqrt{t} \lg n)$ in time $O(mn)$, i.e. linear time in the input size. However, for the hereditary setup, nothing faster than Bansal's result is known. Giving such an algorithm is the main focus of this work.

\paragraph{Our Contribution.}
In this work, we give a fast algorithm for finding low-discrepancy colorings when the hereditary discrepancy of the input matrix $A$ is small. Concretely, we show the following:
\begin{theorem}
  \label{thm:main}
  There is a randomized algorithm that on any $m \times n$ input matrix $A \in \R^{m \times n}$, finds a coloring $x \in
  \{-1,1\}^n$ for which $\disc(A,x) = O(\lg n \cdot \lg^{3/2}m \cdot \herdisc(A))$ in expected $O(mn^2\lg(2 + m/n) + n^3)$ time.
\end{theorem}
Since the input size is $m n$, the running time is only a factor $n \lg(2 + m/n)$ more than linear (at least for $m \geq n$). Comparing this to the $\Omega(m n^{4.5})$ running time of Bansal's algorithm, this is the first practical algorithm with hereditary guarantees. On the downside, we lose a factor $\lg^{3/2}m$ in the quality of the coloring found compared to Bansal's algorithm. To test the practicality of our algorithm, we implemented it and ran experiments on several types of matrices. The algorithm finishes in reasonable time (less than $30$ minutes on a standard laptop) even on matrices of size $4000 \times 4000$ and $10000 \times 2000$. Moreover, despite the logarithmic factors in the theoretical analysis, it returns colorings that have far lower discrepancy than random colorings have (around a factor $10$ smaller discrepancy on structured matrices), even when compared to repeatedly testing random colorings for the same amount of time as spent by our algorithm (around a factor $4$ smaller discrepancy). See Section~\ref{sec:experiments} for details.

\subsection{Technical Contributions and Algorithm Overview}
\label{sec:overview}
The main obstacle faced in the design of the new algorithm, is the current lack of structural understanding of matrices with low hereditary discrepancy. In more detail, the previous algorithm by Bansal only uses that $A$ has low hereditary discrepancy to guarantee the feasibility of a carefully designed SDP. As such, the insight that a concrete SDP is feasible, does not seem to provide any line of attack beyond solving an SDP. Our first technical contribution is thus to demonstrate that matrices $A$ with low hereditary discrepancy have very stringent requirements on the eigenvalues of $A^TA$. A connection between eigenvalues and hereditary discrepancy was earlier observed by Larsen~\cite{herel2} (other, slightly weaker/incomparable connections were also known~\cite{factNorms,chazelleLvov}) who proved:
\begin{theorem}[Larsen~\cite{herel2}]
  \label{thm:disceigen}
  For an $m \times n$ real matrix $A$, let $\lambda_1 \geq \lambda_2 \geq \cdots \geq \lambda_n \geq 0$ denote the eigenvalues of $A^TA$. For all positive integers $k \leq \min\{n,m\}$, we have
 \[
   \herdisc(A) \geq \frac{k}{2e} \sqrt{\frac{\lambda_k}{mn}}.
 \]
\end{theorem}
In Larsen's work, the same lower bound was also proved for the $\ell_2$-version of hereditary discrepancy. Said briefly, the $\ell_2$-discrepancy of a matrix $A \in \R^{m \times n}$ is defined as $\disc_2(A) = \min_{x \in \{-1,1\}^n} \|Ax\|_2/\sqrt{m}$ and the $\ell_2$-hereditary discrepancy is again defined as the maximum over submatrices of $A$. For $\ell_2$-discrepancy, this connection showed a clear path towards finding low discrepancy colorings. At a high level, the idea is to find a coloring $x \in \{-1,1\}^n$ that is orthogonal to the $n/2$ eigenvectors corresponding to the largest eigenvalues of $A^TA$. For such $x$, it holds that $\|Ax\|_2^2 = x^TA^TAx \leq \lambda_{n/2} \cdot n$, which by the above relation implies $\|Ax\|_2 = O(\sqrt{\lambda_{n/2} \cdot n}) = O(\herdisc_2(A) \sqrt{m})$. By the definition $\disc_2(A,x) = \|Ax\|_2/\sqrt{m}$, this gives a discrepancy of $O(\herdisc_2(A))$. This argument assumed that it was possible to find a coloring $x$ orthogonal to the top $n/2$ eigenvectors. This is not quite the case, resulting in an additional $\sqrt{\lg n}$ factor in the $\ell_2$-discrepancy of Larsen's algorithm by using the classic \emph{partial coloring} technique. We will not discuss this further here.

For the more central $\ell_\infty$-discrepancy, the above connection to eigenvalues and eigenvectors seems very hard to exploit. Concretely, the mere fact that $\|Ax\|_2^2$ is small, says nothing interesting about $\|Ax\|_\infty$, and it is completely unclear that simple bounds on the eigenvalues of $A^TA$ may be useful for finding an $x \in \{-1,1\}^n$ with small $\disc(A,x)$. One of our main technical contributions, is to show that the connection to eigenvalues in Theorem~\ref{thm:disceigen} may be exploited for $\ell_\infty$-discrepancy as well. Concretely, we first observe that if we define $V$ as the matrix having the $n/2$ eigenvectors of $A^TA$ with largest eigenvalues as rows, then $A(I-V^TV)$ cannot have more than $m/2$ rows of norm exceeding $O(\herdisc(A))$. The reason for this, is that the rows of $B=A(I-V^TV)$ are the projections of the rows of $A$ onto the orthogonal complement of the top $n/2$ eigenvectors of $A^TA$ and therefore it must be the case that $B^TB$ has all eigenvalues bounded by $\lambda_{n/2}$. But $\tr(B^TB)$ is equal to the sum of squared row norms of $B$. If there were more than $m/2$ rows with norm more than $c \cdot \herdisc(A)$ for a large constant $c$, then we would have $\tr(B^TB) \geq (c/2) m \herdisc^2(A)$. But $\tr(B^TB)$ is the sum of the eigenvalues of $B^TB$. Since $B^TB$ has rank at most $n$, this implies that $B^TB$ has an eigenvalue of at least $cm \herdisc^2(A)/(2n)$. By Theorem~\ref{thm:disceigen}, this is greater than $\lambda_{n/2}$ for $c$ big enough, i.e. a contradiction.

With this observation in mind, a natural idea would be to find a coloring $x$ in the orthogonal complement of the span of the rows of $V$. Such an $x$ would behave well for the rows of norm $O(\herdisc(A))$, but not for the remaining up to $m/2$ rows of larger norm. To get around this, we instead pick only the top $n/(8\lg(8m/n))$ eigenvectors of $A^TA$ when forming $V$. Then we can again argue that there is at most $m/2$ rows of $B = A(I-V^TV)$ with norm exceeding $O(\herdisc(A) \lg(8m/n))$. We then focus on the submatrix $B'$ of $B$ obtained by deleting all but the $m/2$ rows of largest norm. We then find the top $n/(8\lg(8m/n))$ eigenvectors of $B'^TB'$ and add those as rows of $V$. We can then argue that this reduces the number of rows of $A(I-V^TV)$ of norm more than $O(\herdisc(A) \lg(8m/n))$ to $m/4$. Repeating this, where we delete all but the $m/2^i$ rows of largest norm from $B$, for $\lg(8m/n)$ rounds, leaves less than $n/8$ rows of norm exceeding $O(\herdisc(A) \lg(8m/n))$. We can finally add a basis for the subspace spanned by those $n/8$ rows to $V$. This results in a matrix $V$ with $n/4$ rows, such that every row of $B = A(I-V^TV)$ has norm at most $O(\herdisc(A) \lg(8m/n))$. That is, we have shown the following theorem capturing a key structural property of matrices with low hereditary discrepancy:
\begin{theorem}
  \label{thm:structure}
For any $m \times n$ real matrix $A$ with $m \geq n$, there is an $(n/4) \times n$ matrix $V$, having unit length orthogonal rows, such that all rows of $A(I-V^TV)$ have norm at most $O(\herdisc(A) \lg(2m/n))$. Moreover, such a matrix $V$ can be computed in $O(mn^2\lg(2 + m/n))$ time.
\end{theorem}
We remark that a similar result could be derived from the $\gamma_2$ norm characterization of hereditary discrepancy~\cite{gamma2}. Concretely, it can be shown that a matrix $A$ has a factorization $A = BC$ so that all columns of $C$ have $\ell_2$ norm at most $1$ and all rows of $B$ have $\ell_2$ norm at  most $\gamma_2(A) = O(\herdisc(A) \lg \rank(A))$. Then $C^TC$ has at most $n/2$ eigenvalues larger than $2$. We could then let $V$ have the top $n/2$ eigenvectors of $C^TC$ has rows, implying that $C(I-V^TV)$ has operator norm at most $\sqrt{2}$ and therefore all rows of $A(I-V^TV) = BC(I-V^TV)$ have norm bounded by $\sqrt{2} \gamma_2(A) = O(\herdisc(A) \lg \rank(A))$. Computing the factorization can be done by solving an SDP, which would be slower than our approach.

Having computed the matrix $V$ in Theorem~\ref{thm:structure}, we now seek a ``random'' coloring $x$ in the orthogonal complement of the rows of $V$. Our basic approach for finding such a coloring, is to run the Edge-Walk algorithm by Lovett and Meka~\cite{lovett}. Here we start out with $x = 0$ and then repeatedly sample a random vector $g$ in the orthogonal complement and add it to $x$. When a coordinate $x_i$ reaches $1$, we add $e_i$ as a row of $V$, ensuring that all further $g$ have $g_i = 0$ and thus leaves the coordinate unchanged. Moreover, whenever $|\langle a_i, g \rangle|$ exceeds $O(\herdisc(A) \lg^{3/2}(8m/n))$, we add $a_i$ as a row of $V$, ensuring that $\langle a_i, g \rangle$ remains unchanged. This is similar in spirit to the Edge-Walk algorithm, with the key new observation being that the change in $|\langle a_i, g \rangle|$ is proportional to the length of the projection of $a_i$ onto the orthogonal complement of the rows of $V$, and thus proportional to $\herdisc(A)$. On top of this, we pay a logarithmic factor due to partial coloring, which is similar to most previous discrepancy minimization algorithms.

As a second technical contribution, we observe that the Edge-Walk algorithm would need $\Omega(n^2)$ steps of sampling a random $g$ (for the familiar reader, the Edge-Walk algorithm would need to add $e_i$ to $V$ when $|x_i|$ exceeds $1-1/n$ and this would require a small step size to avoid $|x_i|$ exceeding $1$). We carefully reduce this to only $O(n)$ steps by modifying the update performed in each step of the algorithm. The Edge-Walk algorithm always adds $\eta g$ to $x$ for an $\eta$ small enough that no coordinate of $x_i$ ever changes from less than $1-1/n$ in absolute value to above $1$. Then upon termination, each coordinate of $x$ may be rounded to either $-1$ or $+1$. This approach requires a very small step size $\eta$. What we do instead, is that we compute the largest possible $\mu>0$, such that adding $\mu g$ or $-\mu g$ to $x$ ensures that no coordinate $x_i$ exceeds $1$ in absolute value. We then add $\min\{\eta, \mu\}g$ to $x$. At a high level, this allows the Edge-Walk algorithm to temporarily violate $|x_i| \leq 1$, but when it happens, the step size is slightly reduced in that iteration, making $|x_i|=1$ instead. By a careful analysis, this reduces the number of iterations to just $O(n)$.

Let us finally comment on an alternative approach for obtaining a result similar to ours. In the work by Dadush et al.~\cite{balancing}, it is shown that running an algorithm by Rothvoss~\cite{rothvoss} with the convex body $K = \{x : \|Ax\|_\infty \leq 1\}$ can obtain a coloring whose discrepancy is bounded by the hereditary discrepancy of $A$. To run Rothvoss' algorithm, one needs to projet vectors onto $K$. This can be done by solving a linear program~\cite{Eldan2018EfficientAF}.

In Section~\ref{sec:prelim}, we introduce some inequalities and facts needed for the analysis of our algorithm and then proceed in Section~\ref{sec:algo} to present our new algorithm and its analysis.

\section{Preliminaries}
We will need a few inequalities regarding normal distributed random
variables for the analysis of our algorithm. The first is a standard
fact that we state without proof:
\label{sec:prelim}
\begin{claim}
  \label{claim:norm}
  Let $G \sim \Norm(0,1)$. Then, for any $\lambda > 0$, $\Pr[|G| \geq
  \lambda] \leq 2\exp(-\lambda^2/2)$.
\end{claim}
Secondly, we need the following fact regarding the distribution of the
inner product between a vector $a$ and the projection of a vector with i.i.d.
normal distributed entries onto a subspace:
\begin{claim}
  \label{claim:ipprojnormal}
  Let $V \in \R^{\ell \times n}$ be a matrix with $\ell \leq n$ orthogonal
  unit length rows. Let $g \in \R^n$ be sampled with $n$
  i.i.d. $\Norm(0,1)$ distributed entries and let $a \in \R^n$ be an
  arbitrary vector. Then
  \[
    \langle a, (I-V^TV)g \rangle \sim \Norm(0,\|a(I-V^TV)\|^2).
  \]
\end{claim}
A proof of Claim~\ref{claim:ipprojnormal} can be found in
Section~\ref{sec:boring}.

Finally, we need a version of Azuma's inequality for Martingales with
Subgaussian tails.
\paragraph{Azuma for Martingales with Subgaussian Tails.}
A sequence of random variables $Z_1,Z_2,\dots$ is called a
\emph{martingale difference sequence} with respect to another sequence
of random variables $X_1,X_2,\dots$, if for any $t$, $Z_{t+1}$ is
measurable wrt. the sigma algebra generated by $X_1,\dots,X_t$, and $\E[Z_{t+1} \mid X_1, \dots, X_t] =
0$ with probability $1$.

\begin{theorem}[Shamir~\cite{azuma}]
  \label{thm:azuma}
  Let $Z_1,\dots,Z_T$ be a martingale difference sequence with respect
  to a sequence $X_1,\dots,X_T$, and suppose there are constants
  $b>1$, $c>0$ such that for any $t$ and any $a > 0$, it holds
  that
  \[
    \max\{\Pr[Z_t > a\mid X_1,\dots,X_{t-1}], \Pr[Z_t < -a \mid
    X_1,\dots,X_{t-1}]\} \leq b \exp(-ca^2).
  \]
  Then for any $\delta > 0$, it holds with probability at least
  $1-\delta$ that
  \[
    \left|\frac{1}{T}\sum_{t=1}^T Z_t \right|\leq 2\sqrt{\frac{28 b \lg(1/\delta)}{cT}}.
  \]
\end{theorem}
We remark that~\cite{azuma} proves the theorem without absolute values on the sum
of $Z_t$'s, and also without the factor $2$ on the right hand
side. However, defining $Z_i = -Z_i$ for all $i$, symmetry and a union
bound over the original martingale difference sequence and the negated
one implies the above.

\section{Algorithm}
\label{sec:algo}
In this section, we present our new algorithm with hereditary guarantees. As discussed in Section~\ref{sec:overview}, one of its key ingredients is a new structural property of matrices with low hereditary discrepancy (Theorem~\ref{thm:structure}). We start in Section~\ref{sec:structure} by proving Theorem~\ref{thm:structure}. We then proceed in Section~\ref{sec:hereditary} to present our algorithm for discrepancy minimization with hereditary guarantees. Throughout this section, we assume $m \geq n$. The case of $m < n$ can be reduced to $m = n$ via standard techniques in time $O(mn^2 + n^3)$. For the interested reader, we have included a sketch of this reduction in Section~\ref{sec:boring}. This reduction accounts for the $O(n^3)$ term in the running time stated in Theorem~\ref{thm:main}.

\subsection{Structure of Matrices with Low Hereditary Discrepancy}
\label{sec:structure}
The goal of this section is to prove Theorem~\ref{thm:structure}, i.e. to give an algorithm that on an $m \times n$ matrix $A$ (with $m \geq n$), outputs an $\ell \times n$ matrix $V$, with $\ell \leq n/4$, such that all rows of $V$ are orthogonal unit length vectors and all rows of $A(I-V^TV)$ have norm $O(\herdisc(A) \lg(2m/n))$. The algorithm is based on the ideas outlined in Section~\ref{sec:overview} and uses as a subroutine the simple algorithm presented as Algorithm~\ref{alg:orth}. Algorithm~\ref{alg:orth} takes as input an $\ell \times n$ matrix $V$ assumed to have unit length orthogonal rows, as well as another vector $s \in \R^n$. It then computes $s(I-V^TV)$, which is the projection of $s$ onto the orthogonal complement of the rows of $V$. If this vector is non-zero, it is scaled to unit length and added as a row of $V$. We will use Algorithm~\ref{alg:orth} as a subroutine to build matrices $V$ with orthogonal unit length rows. Concretely, invoking Algorithm~\ref{alg:orth} with an input vector $s$, guarantees that $s$ afterwards lies in the span of the rows of $V$ and that $V$ continues to have unit length orthogonal rows. It is nothing more than the Gram-Schmidt process.
\begin{algorithm}
  \DontPrintSemicolon
  \KwIn{Vector $s \in \R^{n}$, Matrix $V \in \R^{\ell \times n}$ with $\ell$ rows forming an orthonormal basis}
  \KwResult{Matrix $V' \in \R^{\ell' \times n}$ with at most $\ell+1$ rows forming an orthonormal basis}
  Let $s' = s(I-V^TV)$.

      \If{$s' \neq 0$}{
        Add $s'/\|s'\|$ as a row of $V$.
        }
  \Return{$V$}
        
  \caption{Orthogonalize}\label{alg:orth}
\end{algorithm}

Using Algorithm~\ref{alg:orth} as a subroutine, we are ready to present our algorithm for computing the matrix $V$ with the guarantees claimed in Theorem~\ref{thm:structure}. The algorithm is presented as Algorithm~\ref{alg:project}. As discussed in Section~\ref{sec:overview}, the idea is to build $V$ iteratively. In the $i$'th step, all but the $m/2^i$ rows of $A(I-V^TV)$ of largest norm are deleted. The eigenvalues and eigenvectors of the resulting matrix are then computed and the top $O(n/\lg(m/n))$ eigenvectors are added to $V$.

\begin{algorithm}
  \DontPrintSemicolon
  \KwIn{Matrix $A \in \R^{m \times n}$}
  \KwResult{Matrix $V \in \R^{\ell \times n}$ with $\ell \leq n/4$ rows forming an orthonormal basis}

  $V = 0 \in \R^{0 \times  n}$ \tcp*{Initially empty matrix}
  
  \For(){$i=1,\dots,\lg_2(8m/n)$}{

    Let $B =A(I-V^TV)$.

    Let $\bar{B}$ be the submatrix obtained from $B$ by deleting all but the $m/2^{i-1}$ rows of largest norm.

    Compute the eigenvalues $\mu_1 \geq \cdots \geq \mu_n \geq 0$ and corresponding eigenvectors $\eta_1,\dots,\eta_n$ of $\bar{B}^T\bar{B}$.

    Add $\{\eta_1,\dots,\eta_{n/(8 \lg_2(8m/n))} \}$ as rows of $V$.
    }
    Let $B =A(I-V^TV)$.

    Let $r_1,\dots,r_{n/8}$ be the $n/8$ rows of $B$ with largest norm.

    \For(){$j=1,\dots,n/8$}{
      Orthogonalize($r_j$,$V$) \tcp*{Add $r_j$ as row of $V$ via Algorithm~\ref{alg:orth}}
   }
  \Return{$V$}
        
  \caption{ProjectToSmallRows}\label{alg:project}
\end{algorithm}

We prove the following guarantees on the output of Algorithm~\ref{alg:project} and remark that Theorem~\ref{thm:structure} is an immediate corollary of Lemma~\ref{lem:project} and Lemma~\ref{lem:runtime} below.

\begin{lemma}
  \label{lem:project}
Let $V \in \R^{\ell \times n}$ be the matrix returned by Algorithm~\ref{alg:project} on input matrix $A \in \R^{m \times n}$ with $m \geq n$. Then $\ell \leq n/4$ and every row $a_i$ of $A(I-V^TV)$ has norm no more than $48 e \lg_2(8m/n) \herdisc(A)$.
\end{lemma}

\begin{proof}
We start by arguing that the rows of $V$ form an orthonormal basis. This is clearly true initially as $V$ has no rows. Now consider adding the eigenvectors $\{\eta_1,\dots,\eta_{n/(8 \lg_2(8m/n))} \}$ as rows of $V$ in line 6 of Algorithm~\ref{alg:project}. These are orthogonal to each other and have unit length. Moreover, they reside in the span of the rows of $\bar{B}$ and hence also in the span of the rows of $B$. Since all rows of $B$ are orthogonal to all rows in $V$ (since $B=A(I-V^TV)$), it follows that adding the eigenvectors to $V$ maintains that $V$ is an orthonormal basis. Finally, the last for-loop clearly preserves that the rows of $V$ form an orthonormal basis.

We now prove by induction that upon completing the $i$'th iteration of the first for-loop (lines $2$ to $6$), there are no more than $m/2^i$ rows in $A(I-V^TV)$ of norm exceeding $48 e \lg_2(8m/n) \herdisc(A)$. For $i=0$ this trivially holds. For the inductive step, let $i \in \{1,\dots,\lg_2(8m/n)\}$ and assume that after iteration $i-1$, there are no more than $m/2^{i-1}$ rows of norm more than $48 e \lg_2(8m/n) \herdisc(A)$. For the $i$'th iteration, this implies that all rows not in $\bar{B}$ have norm at most $48 e \lg_2(8m/n) \herdisc(A)$. Since the norm of any row of $A(I-V^TV)$ may only decrease by adding more rows to $V$, we conclude that this remains the case for rows not in $\bar{B}$. What is left to show, is that adding $\{\eta_1,\dots,\eta_{n/(8 \lg_2(8m/n))} \}$ as rows of $V$, leaves at most $m/2^i$ rows in $\bar{B}(I-V^TV) \in \R^{m/2^{i-1} \times n}$ with norm exceeding $48 e \lg_2(8m/n) \herdisc(A)$ (before adding the new rows to $V$ in line $6$ of iteration $i$, it holds that $\bar{B} = \bar{B}(I-V^TV)$, as all rows of $\bar{B}$ are already orthogonal to all rows of $V$).

Assume for the sake of contradiction that there are more than $m/2^i$ rows in $\bar{B}(I-V^TV)$ with norm exceeding $48 e \lg_2(8m/n) \herdisc(A)$ after adding $\{\eta_1,\dots,\eta_{n/(8 \lg_2(8m/n))} \}$ as rows of $V$ in line 6 of iteration $i$ of the first for-loop. Pick an arbitrary orthonormal basis $u_1,\dots,u_{n-\ell}$ for the orthogonal complement $C$ of the rows of $V$ (where $\ell$ is the number of rows of $V$). Then all rows $b_i$ of $\bar{B}(I-V^TV)$ lie in the span of $u_1,\dots,u_{n-\ell}$. It follows that $\sum_i \sum_j \langle b_i, u_j\rangle^2 = \sum_i \|b_i\|^2 > (m/2^i) (48 e \lg_2(8m/n) \herdisc(A))^2$. Averaging over all $u_j$, there must exists a $j$ with $\sum_i \langle b_i, u_j \rangle^2 > (m/2^i) (48 e \lg_2(8n/m) \herdisc(A))^2/n$. Thus we have a unit vector $v=u_j$ in $C$ with
\[
  \|\bar{B}v\|^2 > (m/2^i) (48 e \lg_2(8m/n) \herdisc(A))^2/n.
\]
Since $v$ is orthogonal to $\eta_1,\dots,\eta_{n/(8 \lg_2(8m/n))}$, we also have that $\|\bar{B}v\|^2 = v \bar{B}^T\bar{B} v \leq \mu_{n/(8 \lg_2(8m/n))}$. Combining the two yields
  \[
    \mu_{n/(8 \lg_2(8m/n))} > (m/2^i) (48 e \lg_2(8m/n) \herdisc(A))^2/n.
  \]
Now let $\bar{A}$ be the submatrix of $A$ obtained by deleting the same rows as when choosing $\bar{B}$. Then $\herdisc(A) \geq \herdisc(\bar{A})$. What remains is to relate $\herdisc(\bar{A})$ to the eigenvalues of $\bar{B}$. Recall that $\bar{B} = \bar{A}(I-V'^TV')$ where $V'$ is the matrix $V$ at the beginning of the loop iteration (before adding eigenvectors $\eta_j$ in step 6). We may assume that all the eigenvectors $\eta_j$ corresponding to non-zero eigenvalues of $\bar{B}^T\bar{B}$ are orthogonal to the rows of $V'$. Hence $\bar{B}\eta_j = \bar{A}\eta_j$. For any $k$, the vectors $\eta_1,\dots,\eta_k$ are thus orthogonal and have $\eta_j^T \bar{A}^T \bar{A} \eta_j \geq \mu_k$. It follows by the min-max characterization of eigenvalues that the $k$'th largest eigenvalue of $\bar{A}^T\bar{A}$ is at least $\mu_k$. Moreover, Theorem~\ref{thm:disceigen} gives us, with $k = n/(8 \lg_2(8m/n))$, that
  \[
    \herdisc(A) \geq \herdisc(\bar{A}) > \frac{n}{16e \lg_2(8m/n)} \sqrt{\frac{(m/2^i) (48 e \lg_2(8m/n) \herdisc(A))^2/n}{(m/2^{i-1})n}} = \sqrt{2} \herdisc(A),
  \]
  which is a contradiction.

  The induction proof implies that after completing the last iteration, $i = \lg_2(8m/n)$, of the first for-loop, there are no more than $m/2^i = n/8$ rows in $A(I-V^TV)$ of norm more than $48e \lg_2(8m/n) \herdisc(A)$. Since all such rows are in the span of the rows of $V$ after completing the second for-loop, we conclude that the algorithm has the claimed guarantees.
\end{proof}

\begin{lemma}
  \label{lem:runtime}
  Algorithm~\ref{alg:project} can be implemented to run in $O(mn^2\lg(2+m/n))$ time.
\end{lemma}

\begin{proof}
In each step of the first for-loop, computing the matrix $B$ can be done in $O(mn^2)$ time. Constructing $\bar{B}$ takes $O(mn)$ time for computing norms (and $O(m)$ time for finding the largest $m/2^{i-1}$ of them). Computing $\bar{B}^T \bar{B}$ takes $O(n^2m/2^i)$ time and computing an eigendecomposition of $\bar{B}^T \bar{B}$ takes $O(n^3)$ time. There are $O(\lg(2 + m/n))$ iterations of the for-loop, for a total running time of $O((mn^2 + n^3) \lg(2 + m/n))$. The second for-loop has $O(n)$ iterations, each taking $O(n^2)$ time when calling Orthogonalize. Hence we conclude that the total running time is $O((mn^2 + n^3) \lg(2 + m/n)) = O(mn^2 \lg(2 + m/n))$.
\end{proof}

Lemma~\ref{lem:project} and Lemma~\ref{lem:runtime} together imply Theorem~\ref{thm:structure}.

\subsection{Discrepancy Minimization with Hereditary Guarantees}
\label{sec:hereditary}
In this section, we give our new algorithm for discrepancy minimization with hereditary guarantees. The algorithm follows the outline in Section~\ref{sec:overview}, combined with the \emph{partial coloring} technique. In more detail, given an input $m \times n$ matrix $A$, we initialize a coloring $x = 0 \in \R^n$. We then repeatedly update $x$ such that more and more entries become either $-1$ or $1$. In each iteration, the number of coordinates $x_i$ with $|x_i| < 1$ halves, while no coordinate ever exceeds $1$ in absolute value. While changing the coordinates of $x$, we ensure that $\|Ax\|_\infty$ changes by as little as possible.

To implement the above strategy, the algorithm presented further below as Algorithm~\ref{alg:partial} (PartialColoring), takes as input a matrix $A \in \R^{m \times n}$ and a partial coloring $x \in (-1,1)^n$. It then returns another vector $x' \in [-1,1]^n$ such that at least $n/2$ coordinates $x'_i$ are in $\{-1,1\}$, and moreover, $|\|Ax'\|_\infty - \|Ax\|_\infty| \leq O(\lg^{3/2}(2m/n) \cdot \herdisc(A))$. This is stated formally here:
\begin{theorem}
  \label{thm:partial}
For any input matrix $A \in \R^{m \times n}$ with $m \geq n$, and any partial coloring $x \in (-1,1)^n$, it holds with probability at least $0.08$ that Algorithm~\ref{alg:partial} (PartialColoring) on input $(A,x)$, returns a vector $x'$ for which $\|x'\|_\infty = 1$, $|\{i : |x'_i|=1\}| \geq n/2$ and for all rows $a_i$ of $A$, it holds that $|\langle a_i, x' \rangle - \langle a_i,x\rangle| \leq O(\lg^{3/2}(2m/n) \cdot \herdisc(A))$. With the remaining probability, Algorithm~\ref{alg:partial} returns \textbf{Failure}. Moreover, Algorithm~\ref{alg:partial} can be implemented to run in expected $O(mn^2 \lg(2 + m/n))$ time.
\end{theorem}
Using Algorithm~\ref{alg:partial} as a subroutine, our final discrepancy minimization algorithm is presented as Algorithm~\ref{alg:hereditary} (\Minimize).

\begin{algorithm}
  \DontPrintSemicolon
  \KwIn{Matrix $A \in \R^{m \times n}$}
  \KwResult{Coloring $x \in \{-1,1\}^n$}
  
  Let $x \gets 0 \in \R^n$.
  
  \While{$x \notin \{-1,1\}^n$}{
    Let $S \subseteq [n]$ be the subset of indices $i$ such that
    $|x_i|<1$.

    Let $\bar{x} \gets x_S$ be the coordinates of $x$ indexed by $S$.

    Let $\bar{A} \gets A_S$ be the columns of $A$ indexed by $S$.

    Finished $\gets$ False
    
    \While{not Finished}{
      $x' \gets $PartialColoring($\bar{A},\bar{x}$). \tcp*{Call
        Algorithm~\ref{alg:partial}}
      \If{$x' \neq$\textbf{Failure}}{
        $x_S \gets x'$. \tcp*{Replace coordinates in $x$ indexed by $S$, by
      the new vector $x'$}
        Finished $\gets$ True
      }
    }

  }
  \Return{$x$}
        
  \caption{\Minimize}\label{alg:hereditary}
\end{algorithm}

Let us see that Algorithm~\ref{alg:hereditary} (\Minimize) gives the guarantees stated in our main result, Theorem~\ref{thm:main}. For the bound on the discrepancy of the returned coloring $x$, observe that Algorithm~\ref{alg:hereditary} has at most $\lg_2 n$ iterations of the outer while loop, since each iterations halves the number of entries $x_i$ with $|x_i| \notin \{-1,1\}$ (by Theorem~\ref{thm:partial}). In the $i$'th iteration, we have $|S| \leq n/2^i$. Thus replacing the entries in $x_S$ by $x'$ in line 10 increases $\|Ax'\|_\infty$ by at most an additive $O(\lg^{3/2}(m/|S|) \herdisc(\bar{A})) = O(\lg^{3/2}m \herdisc(A))$ compared to $\|Ax\|_\infty$ (here we use that $\herdisc(\bar{A}) \leq \herdisc(A)$). Thus the final coloring $x$ satisfies $\disc(A,x) = O(\lg n \cdot \lg^{3/2}m \cdot \herdisc(A))$ as claimed in Theorem~\ref{thm:main}.

As an interesting remark, notice that if we define $\herdisc(A,k)$ is the maximum discrepancy over all $m \times k$ submatrices of $A$, then the discrepancy of the returned coloring is bounded by
\[
  O\left(\sum_{i=0}^{\lg_2 n} \lg^{3/2}(m2^{i+1}/n) \cdot \herdisc(A,n/2^i)\right).
\]
This may be an improvement for some matrices $A$. For instance in the case of geometric set systems, it is often the case that $\herdisc(A,k)$ grows as a polynomial in $k$. In that case, the sum is dominated by the first term, $i=0$, and the discrepancy improves to $O(\lg^{3/2}(2m/n) \herdisc(A))$.

For the running time, notice that if we compute the set $S$ for each iteration of the outer while loop by examining only the coordinates $x_i$ with $i \in S$ from the previous iteration, then $S$ can be maintained in total time $O(n)$ throughout the execution of the algorithm ($|S|$ halves with each iteration). Extracting the submatrices $\bar{A}$ takes a total of $O(mn)$ time by the same argument. Finally, in the $i$'th iteration of the outer while loop, we have $|S| \leq n/2^i$. Hence by Theorem~\ref{thm:partial}, the call in line 8 to PartialColoring (Algorithm~\ref{alg:partial}) takes $O(m (n/2^i)^2\lg(2 + 2^i m/n)) = O(mn^2\lg(2 + m/n)2^{-i})$ time in expectation. Moreover, the expected number of times it needs to be called before not returning \textbf{Failure}, is at most $1/0.08 = O(1)$. Hence the total running time for all calls to PartialColoring is $O(mn^2 \lg(2 + m/n))$ in expectation.

What remains is thus to specify the PartialColoring algorithm (Algorithm~\ref{alg:partial}). As mentioned in Section~\ref{sec:overview}, the basic idea is to use the structural properties of matrices with low hereditary discrepancy in order to find a matrix $V$ for which $A(I-V^TV)$ has rows of small norm. From there on, we basically execute the Edge-Walk algorithm by Lovett and Meka, except that we optimized it a bit by taking large step sizes. This is accomodated by capping the magnitude of steps taken, such that no coordinate of $x_i$ ever exceeds $1$ in absolute value (see line 8-9). The algorithm is presented in details here:

\begin{algorithm}
  \DontPrintSemicolon
  \KwIn{Matrix $A \in \R^{m \times n}$\\
    \quad Partial coloring $x \in (-1,1)^n$}
  \KwResult{Vector $x \in \R^n$}

  $v \gets 0 \in \R^n$.

  $V \gets$ ProjectToSmallRows($A$) \tcp*{Call Algorithm~\ref{alg:project}}

  $\eta \gets \max_i \|a_i(I-V^TV)\|$ \tcp*{Largest norm of a row in
    $A(I-V^TV)$}

  $\eps \gets (\max\{4 \ln(mn) + 20,256n\})^{-1/2}$, $Q \gets 16\eps^{-2} + 256 n$, $\tau \gets 22 \eps \eta \sqrt{Q \lg(256m/n)}$
  
  \For(){$t=1,\dots,Q$}{
      
      Sample $g$ with $n$ i.i.d. $\Norm(0,1)$ entries.

      $g \gets g(I-V^TV)$. \tcp*{Sample $g$ in orthogonal complement
        of rows of $V$}

      Let $\mu > 0$ be maximal such that $\max\{\|x + v + \mu g\|_\infty, \|x + v -
      \mu g\|_\infty\} =1$.

      $g \gets \min\{\eps,\mu\} g$.

      \For(){every $i$ such that $|x_i + v_i + g_i| =1$ and $|x_i +
        v_i| < 1$}{
        Orthogonalize($e_i$,$V$) \tcp*{Add $e_i$ as row of $V$ via Algorithm~\ref{alg:orth}}
        }

        \For(){every row $a_i$ of $A$ such that $|\langle a_i,
          v+g\rangle| \geq \tau$ and $|\langle a_i,v\rangle| <
          \tau$}{
          \If(){$|\langle a_i, v+g \rangle| > \tau + \eta$}{
            \Return{\textbf{Failure}}
            }\Else{
          Orthogonalize($a_i$,$V$) \tcp*{Add $a_i$ as row of $V$ via
            Algorithm~\ref{alg:orth}}
          }
        }
        
        $v \gets v + g$.

        \If(){$|\{i : |x_i + v_i| = 1\}| \geq n/2$}{
              \Return{$x + v$}
        }
      }
    
  \Return{\textbf{Failure}}
        
  \caption{PartialColoring}\label{alg:partial}
\end{algorithm}

First we show that Algorithm~\ref{alg:partial}, on an input $x \in (-1,1)^n$, rarely fails in line 14:
\begin{lemma}
  \label{lem:smalldisc}
On any input $A \in \R^{m \times n}$ and any vector $x \in (-1,1)^n$, the probability that Algorithm~\ref{alg:partial} returns \textbf{Failure} in line 14 is at most $1/86$.
\end{lemma}

\begin{proof}
  We need to bound $\Pr[|\langle a_i, v + g \rangle|>\tau +  \eta]$. Examining Algorithm~\ref{alg:partial}, we observe that if $|\langle a_i, v + g \rangle| \geq \tau$ in line 12, we add $a_i$ as a row of $V$ or abort with \textbf{Failure}. If we add $a_i$ as a row, this implies that all subsequent sampled vectors $g$ will be orthogonal to $a_i$ due to line 7. Hence the only way it could happen that $|\langle a_i, v + g \rangle| \geq \tau + \eta$, is, if during some iteration of the for-loop, the final vector $g$ in line 9 results in $|\langle a_i, v + g \rangle| > \tau + \eta$ while $|\langle a_i, v\rangle| < \tau$ (from the previous iteration). By linearity of inner product, this can only happen if $|\langle a_i, g \rangle| > \eta$. Consider line 9 of Algorithm~\ref{alg:partial}. There we set $g$ to $\min\{\eps,\mu\} g$. Hence the probability that $|\langle a_i, g \rangle| > \eta$ is bounded by the probability that a vector $\tilde{g}$ sampled with $n$ i.i.d. $\Norm(0,1)$ entries and then projected to $\tilde{g} \gets \eps \tilde{g}(I - V^TV)$ satisfies $|\langle a_i, \tilde{g} \rangle| > \eta$. The distribution of $\langle a_i, \tilde{g} \rangle$ for such $\tilde{g}$ is $\Norm(0,\eps^2\|a_i(I-V^TV)\|^2)$ by Claim~\ref{claim:ipprojnormal}, i.e. normal distributed with mean $0$ and variance at most $\sigma^2 \leq \eps^2 \eta^2$. This is equal in distribution to $\sigma \Norm(0,1)$ and thus by Claim~\ref{claim:norm}, we have $\Pr[|\langle a_i, \tilde{g} \rangle| > \eta] \leq 2 \exp(-(\eta^2/\sigma^2)/2) \leq 2 \exp(-\eps^{-2}/2)$. We now define an event $E_{t,i}$ for every iteration $t$ of the for-loop and every row $i$ which occurs if the vector $g$ sampled in lines 6-7 during iteration $t$ and then scaled by $\eps$ satisfies $|\langle a_i, g \rangle| > \eta$. The event $E_{t,i}$ does not occur if Algorithm~\ref{alg:partial} already returned \textbf{Failure} in an iteration $t' < t$. By the previous arguments, we have $\Pr[E_{t,i}] \leq 2 \exp(-\eps^{-2}/2)$. Since $\eps \leq (4 \ln(mn) + 20)^{-1/2}$, this is less than $1/(nm)^3$. A union bound over all $E_{t,i}$ for $i=1,\dots,m$ and all $t=1,\dots,Q$ shows that the probability that Algorithm~\ref{alg:partial} returns a vector $x' = x + v$ such that there is a row $a_i$ with $|\langle a_i, x' \rangle| - |\langle a_i , x \rangle| > \tau + \eta$ is at most $Qm/((nm)^2e^{10}) \leq 1/86$.
\end{proof}

We next want to bound the probability that the algorithm fails in line 20. For this, we first need a handle on the number of rows $a_i$ that are added to $V$ in line 16:

\begin{lemma}
  \label{lem:violate}
  For any input  $A \in \R^{m \times n}$ and any vector $x \in (-1,1)^n$, let $R$ denote the random variable giving the number of rows $a_i$ added to $V$ in line 16 throughout the execution of Algorithm~\ref{alg:partial}. Then $\E[R] \leq n/256$.
\end{lemma}

\begin{proof}
Fix a row $a_i$. Then $a_i$ is added to $V$ only if $|\langle a_i, v \rangle|$ reaches $\tau$ during some iteration of the for-loop. Also, since we then add it to $V$, it will remain in $V$ for all remaining iterations. This implies that all subsequently sampled $g$'s will leave $\langle a_i, v \rangle$ unchanged.
Consider some iteration $t$ of the for-loop and let $g_t$ denote the vector formed in line 9 of Algorithm~\ref{alg:partial} in iteration $t$. If Algorithm~\ref{alg:partial} returns \textbf{Failure} before iteration $t$, define $g_t = 0$. Define also $v_t = \sum_{i=1}^t g_i$. Then $\langle a_i, v_t \rangle = \sum_{j=1}^t \langle a_i, g_j \rangle$. Since $\langle a_i, v\rangle$ remains unchanged if $a_i$ is added to $V$ or if we return \textbf{Failure}, it follows that $|\langle a_i, v_t \rangle| \geq \tau$ during some iteration $t$, only if $|\langle a_i, v_Q \rangle| \geq \tau$. Hence we bound the probability that $|\langle a_i, v_Q \rangle | \geq \tau$.

Observe now that if we condition on $v_{t-1},\dots,v_1$, and let $V_t$ denote the set $V$ at the beginning of iteration $t$ (which is completely determined from $v_{t-1},\dots,v_1$ and the input $A$, $x$), and then draw a vector $\tilde{g}$ with each coordinate i.i.d. $\Norm(0,1)$ distributed, then for any $t>0$, we have $\Pr[|\eps\langle a_i, \tilde{g}^T(I-V_{t}^TV_{t}) \rangle| > t] \geq \Pr[|\min\{\eps, \mu\}\langle a_i, \tilde{g}^T(I-V_{t}^TV_{t}) \rangle| > t] \geq \Pr[|\langle a_i, g_t \rangle| > t]$. We also have that $\eps\langle a_i, \tilde{g}^T(I-V_{t}^TV_{t}) \rangle$ is $\Norm(0, \eps^2 \|a_i(I-V_t^TV_t)\|^2)$ distributed. The variance $\sigma^2$ is at most $\eps^2 \eta^2$, hence by Claim~\ref{claim:norm}, it holds that $\Pr[|\langle a_i, g_t \rangle| > t] \leq \Pr[ |\eps\langle a_i, \tilde{g}(I-V_{t}^TV_{t}) \rangle| > t]  \leq 2\exp(-(\eps \eta)^{-2}t^2/2)$. Finally, observe that $\E[\langle a_i, g_t \rangle \mid v_{t-1},\dots,v_1] = 0$. This follows by symmetry of the distribution of $g_t$ and crucially uses that we define $\mu$ in line 8 to be such that  $\max\{\|x +v+\mu g\|_\infty, \|x + v - \mu g\|_\infty\}=1$. If we had only defined $\mu$ such that $\|x + v + \mu g\|_\infty = 1$ (we add $\mu g$ to $v$, not $-\mu g$), then the distribution of $g$ would not be symmetric.  Thus $\langle a_i, g_1 \rangle,\dots,\langle a_i, g_t\rangle$ forms a martingale difference sequence with respect to $v_1,\dots,v_t$, where the parameter $c$ in Theorem~\ref{thm:azuma} is at least $(\eps \eta)^{-2}/2$ and the parameter $b$ is $2$. Theorem~\ref{thm:azuma} implies that with probability at least $1-\delta$, we have
  \[
    \left| \langle a_i, v_Q \rangle\right| = \left|\sum_{t=1}^Q \langle a_i, g_t \rangle \right| \leq 2 \eps \eta \sqrt{ 112  Q \lg(1/\delta)} \leq 22 \eps \eta \sqrt{Q \lg(1/\delta)}. 
  \]
  Setting $\delta = n/(256 m)$, it follows by linearity of expectation that the expected number of rows added to $V$ throughout the execution of Algorithm~\ref{alg:partial} is $n/256$, i.e. $\E[R] \leq n/256$.
\end{proof}

With Lemma~\ref{lem:violate} established, we are finally ready to bound the probability of returning \textbf{Failure} in line 20:

\begin{lemma}
 \label{lem:fail}
  On any input  $A \in \R^{m \times n}$ and any vector $x \in (-1,1)^n$, the probability that Algorithm~\ref{alg:partial} returns \textbf{Failure} in line 20 is at most $9/10$.
\end{lemma}

\begin{proof}
  The basic idea in the proof is to show that $\E[\|x+v\|^2]$ is large after the $Q$ iterations and thus many coordinates of $x+v$ must reach $1$ in absolute value. With every iteration, we add either $\eps g$ or $\mu g$ to $v$. Since $\mu$ may be the minimum, and thus result in a smaller increase in the norm of $v$, we first bound the probability that $\mu$ is the minimum in line 9.
  Consider line 8-9 of Algorithm~\ref{alg:partial}. If $\mu < \eps$, we have that $\mu$ is maximal for at least one of the constraints $\|x + v +\mu g\|_\infty=1$ and $\|x + v - \mu g\|_\infty=1$. In particular, for line 10 of Algorithm~\ref{alg:partial}, this implies that at least one index $i$ must satisfy $|x_i + v_i + \sigma g|=1$ and $|x_i +v_i|<1$ for a sign $\sigma \in \{-1,1\}$. Let $\rho_t$ denote the probability that Algorithm~\ref{alg:partial} reaches iteration $t$ and that $\mu < \eps$ in line 9 during iteration $t$. Due to the symmetry of $g$ ($g$ and $-g$ are equally likely conditioned on $V$), we get that the expected number of indices $i$ such that $|x_i + v_i + \mu g|=1$ (in iteration $t$) and $|x_i + v_i| < 1$, is at least $\rho_t/2$. Thus the expected number of vectors $e_i$ added to $V$ in line 10-11 throughout the execution of Algorithm~\ref{alg:partial} is at least $\sum_{t=1}^Q \rho_t/2$. Let $X_t$ give the number of $e_i$'s added to $V$ during iteration $t$. Then $\E[\sum_t X_t] \geq \sum_t \rho_t/2$. On the other hand, the sum of $X_t$'s can never exceed $n$, thus we conclude $\sum_t \rho_t \leq 2n$.

Next, define $g_t$ as the vector $g$ added to $v$ in iteration $t$ of Algorithm~\ref{alg:partial} (define it to $g_t = 0$ if the algorithm terminates with \textbf{Failure} before iteration $t$). Define $v = \sum_{i=1}^Q g_t$. If Algorithm~\ref{alg:partial} does not fail, then $v$ is the returned vector. We have:
  \begin{eqnarray*}
    \E[\|x + v\|^2] &=& \E\left[ \left\| x + \sum_{t=1}^Q g_t \right\|^2 \right] \\
    &=& \|x\|^2 + \sum_{i=1}^Q \E[\|g_t\|^2].
  \end{eqnarray*}
  Here the second equality holds because $\E[g_t \mid g_1,\dots, g_{t-1}] = 0$ by symmetry of the distribution of $g_t$. To bound $\E[\|g_t\|^2]$, let $V_t$ be the random variable giving the matrix $V$ at the beginning of iteration $t$ of the for-loop. If the algorithm fails to reach iteration $t$, we define $V_t$ to equal $V$ at the end of the last reached iteration. Let $\tilde{g}_t$ be the random variable giving the value of $g$ during line 6 of iteration $t$, i.e. before projection and scaling by $\min\{\eps,\mu\}$. If the algorithm terminates before iteration $t$, we still define $\tilde{g}_t$ as a random variable with $n$ i.i.d. $\Norm(0,1)$ distributed entries. Define $F_t$ as an indicator taking the value $1$ if the algorithm reaches iteration $t$. Define $Y_t$ as the indicator taking the value $1$ if the algorithm reaches iteration $t$ and $\mu < \eps$ in line 9 during iteration $t$. Then

  \begin{eqnarray*}
    \E[\|g_t\|^2] &=& \E[F_t Y_t \mu^2\|\tilde{g}_t(I-V_t^TV_t)\|^2 + F_t(1-Y_t)\eps^2 \|\tilde{g}_t(I-V_t^TV_t)\|^2] \\
                  &\geq& \E[F_t(1-Y_t)\eps^2 \|\tilde{g}_t(I-V_t^TV_t)\|^2] \\
                  &\geq& \eps^2 \E[ F_t \|\tilde{g}_t(I-V_t^TV_t)\|^2] - \eps^2 \E[Y_t\|\tilde{g}_t(I-V_t^TV_t)\|^2].
  \end{eqnarray*}
  By Cauchy-Schwartz, we have $\E[Y_t \|\tilde{g}_t(I-V_t^TV_t)\|^2] \leq \sqrt{\E[Y_t^2]\E[\|\tilde{g}_t(I-V_t^TV_t)\|^4]}$. Since $Y_t$ is an indicator, we have $\E[Y_t^2] = E[Y_t]$. We also have $\|\tilde{g}_t\| \geq \|\tilde{g}_t(I-V_t^TV_t)\|$. Therefore, it holds that $\E[Y_t \|\tilde{g}_t(I-V_t^TV_t)\|^2] \leq \sqrt{\E[Y_t] \E[\|\tilde{g}_t\|^4]} = \sqrt{\rho_t \E[\|\tilde{g}_t\|^4]}$. Let $c_i$ denote the $i$'th coordinate of $\tilde{g}_t$. We have that the $c_i$'s are i.i.d. $\Norm(0,1)$ distributed. By linearity of expectation:
\begin{eqnarray*}
  \E[\|\tilde{g}_t\|^4] &=& \E\left[ \left(\sum_{i=1}^n c_i^2\right)^2\right] \\
    &=& \sum_{i=1}^n \sum_{j=1}^n \E[ c_i^2 c_j^2].
 \end{eqnarray*}
 For $i \neq j$, we have $\E[c_i^2 c_j^2] = \E[c_i^2]\E[c_j^2] = 1$ by independence. For $i=j$, we have $\E[c_i^4] = 3$ (by well-known bounds on the $4$'th moment of the normal distribution). Thus $\E[\|\tilde{g}_t\|^4] \leq n(n-1) + 3n \leq 2n^2$. We thus have $ \E[Y_t\|\tilde{g}_t(I-V_t^TV_t)\|^2] \leq \sqrt{2 \rho_t} n$.

 At the same time, we have $\E[F_t \|\tilde{g}_t(I-V_t^TV_t)\|^2]  = \Pr[F_t=1] \E[\|\tilde{g}_t(I-V_t^TV_t)\|^2]= \Pr[F_t=1] (\E[n - \dim(V_t)]) = \Pr[F_t=1](n - \E[\dim(V_t)])$. Now recall the definition of the random variable $R$ as the number of rows added to $V$ in line 16 throughout the execution of Algorithm~\ref{alg:partial}. Define also $X$ as the random variable giving the total number of $e_i$'s added to $V$ throughout the execution of Algorithm~\ref{alg:partial}. Then $\dim(V_t) \leq R+X+n/4$ for all $t=1,\dots,Q$. Hence $\E[\dim(V_t)] \leq \E[R] + \E[X] + n/4$. By Lemma~\ref{lem:violate}, we have $\E[R] \leq n/256$, thus $\E[\dim(V_t)] \leq (65/256)n + \E[X]$. Thus we conclude $\E[\|g_t\|^2] \geq \Pr[F_t=1]\eps^2(n-(65/256)n - \E[X]) - \eps^2 \sqrt{2 \rho_t} n$. By Lemma~\ref{lem:smalldisc}, we have $\Pr[F_t=1] \geq 85/86$. Hence $\E[\|g_t\|^2] \geq (85/86) \eps^2(n-(65/256)n - \E[X]) - \eps^2 \sqrt{2 \rho_t} n$.

Assume for now that $\rho_t \leq 1/128$. Then this is at least $(85/86)\eps^2(n - (65/256)n - \E[X] - (86/85)n/8) \geq (85/86)\eps^2(n-(65/256) n - \E[X] - (33/256)n) = (85/86)\eps^2((158/256)n - \E[X])$. If $\rho_t > 1/128$, we simply lower bound $\E[\|g_t\|^2]$ by $0$. Since $\sum_t \rho_t \leq 2n$, we can have at most $256 n$ iterations $t$ for which $\rho_t \geq 1/128$. For the remaining $Q-256n$ iterations, we have $\E[\|g_t\|^2] \geq (85/86)\eps^2((158/256) n - \E[X])$. Therefore $\E[\|x+v\|^2] \geq \|x\|^2 + (Q-256n)(85/86)\eps^2((158/256) n - \E[X]) \geq (Q-256n)(85/86)\eps^2((158/256) n - \E[X]) = 16(85/86) ((158/256) n - \E[X])$. Notice now that Algorithm~\ref{alg:partial} can never produce an output $v$ for which $\|x + v\|_\infty > 1$. Hence it must be the case that $\E[\|x + v\|^2] \leq n$. This implies $n \geq 16 (85/86)((158/256) n - \E[X])$. Rewriting in terms of $\E[X]$, we get $\E[X] \geq (158/256 - (86/85)(1/16))n > 0.553n$. Now define $Z = n - X$. Then $Z$ is a non-negative random variable and $\E[Z] = n - \E[X] \leq 0.447n$. By Markov's inequality, we have $\Pr[Z > n/2] < 2 \cdot 0.447 = 0.894 < 0.9$. Conversely, $\Pr[Z \leq n/2] > 0.1$. This implies $\Pr[X \geq n/2] > 0.1$. When $X \geq n/2$, Algorithm~\ref{alg:partial} is guaranteed to terminate in line 19 and thus we conclude that the probability of returning failure in line 20 is at most $9/10$.
\end{proof}

With these lemmas established, we may finally conclude the proof of Theorem~\ref{thm:partial}.

\begin{proof}[Proof of Theorem~\ref{thm:partial}]
  By Lemma~\ref{lem:smalldisc}, Lemma~\ref{lem:fail} and a union bound, it holds with probability at least $1/10 - 1/86 > 0.08$ that Algorithm~\ref{alg:partial} does not return \textbf{Failure}, and at the same time, every row $a_i$ satisfies $|\langle a_i, x' \rangle - \langle a_i, x \rangle| = |\langle a_i, x'-x\rangle| = |\langle a_i, v\rangle| \leq \tau + \eta$, where $x'$ is the returned vector. We have $Q = O(n + \max\{\ln(nm), n\}) = O(n + \ln m)$ and $\eps = O(\min\{\ln(nm)^{-1/2}, n^{-1/2}\})$. Hence
  \[
    \tau = O(\eta \min\{\ln(nm)^{-1/2}, n^{-1/2}\} \sqrt{(n + \ln m)} \sqrt{\lg(2m/n)}) = O(\eta \sqrt{\lg(2m/n)}).
  \]
The conclusion follows by invoking Lemma~\ref{lem:project} to conclude that $\eta = O(\lg(2m/n) \herdisc(A))$.

  For the running time,  notice that the call to Algorithm~\ref{alg:project} in line 2 takes $O(mn^2\lg(2 + m/n))$ time by Lemma~\ref{lem:runtime}. Line 3 takes $O(mn^2)$ time. There are $Q = O(\eps^{-2} + n) = O(\max\{\ln(mn), n\} + n) = O(n + \lg m)$ iterations of the for-loop. In each iteration, line 6-7 takes $O(n^2)$ time. Line 8 takes $O(mn)$ time. Checking the conditions in line 10, 12 and 18 takes $O(mn)$ time per iteration. Thus ignoring the calls to Orthogonalize, each iteration of the for-loop takes $O(mn + n^2) = O(mn)$ time. Given that there are $O(n + \lg m)$ iterations, the total running time is $O(mn(n + \lg m))$.  The calls to Orthogonalize in line 11 and 16 are executed a total of $O(n)$ times in expectation (by Lemma~\ref{lem:violate}), and any such call takes $O(n^2)$ time. This adds $n^3$ to the running time, which is dominated by $mn^2$ for $m \geq n$. The total running time is thus $O(mn(n + \lg m) + mn^2\lg(2 + m/n)) = O(mn^2 \lg(2 + m/n))$ in expectation.
\end{proof}

\section{Experiments}
\label{sec:experiments}
In this section, we present a number of experiments to test the practical performance of our discrepancy minimization algorithm. We denote the algorithm by \Minimize in the following.  We compare it to two base line algorithms \Sample and \SampleMany. \Sample simply picks a uniform random $\{-1,+1\}$ vector as its coloring. \SampleMany repeatedly samples a uniform random $\{-1,+1\}$ vector and runs for the same amount of time as \Minimize. It returns the best vector found within the time limit. For \Sample, we ran multiple experiments and took the median result (in fact, we use the median across all the runs of \SampleMany when stating the discrepancy of \Sample).

The algorithms were implemented in Python, using NumPy and SciPy for linear algebra operations. All tests were run on a MacBook Pro (16-inch, 2019) running macOS Monterey 12.4. The machine has a 2.6 GHz 6-Core Intel Core i7 and 16GB 2667 MHz DDR4 RAM.

We tested the algorithms on three different classes of matrices: 
\begin{itemize}
\item \textbf{Uniform} matrices: Each entry is uniform random and independently chosen among $-1$ and $+1$.
\item \textbf{2D Corner} matrices: Obtained by sampling two sets $P=\{p_1,\dots,p_n\}$ and $Q = \{q_1,\dots,q_m\}$ of $n$ and $m$ points in the plane, respectively. The points are sampled uniformly in the $[0,1] \times [0,1]$ unit square. The resulting matrix has one column per point $p_j \in P$ and one row per point $q_i \in Q$. The entry $(i,j)$ is 1 if $p_j$ is \emph{dominated} by $q_i$, i.e. $q_i.x > p_j.x$ and $q_i.y > p_j.y$ and it is $0$ otherwise. Such matrices are known to have hereditary discrepancy $O(\lg^{2.5}n)$~\cite{srinivasan,larsenDisc}.
\item \textbf{2D Halfspace} matrices: Obtained by sampling a set $P = \{p_1,\dots,p_n\}$ of $n$ points in the unit square $[0,1] \times [0,1]$, and a set $Q$ of $m$ halfspace. Each halfspace in $Q$ is sampled by picking one point $a$ uniformly on either the left boundary of the unit square or on the top boundary, and another point $b$ uniformly on either the right boundary or the bottom boundary of the unit square. The halfspace is then chosen uniformly to be either everything above the line through $a,b$ or everything below it. The resulting matrix has one column per point $p_j \in P$ and one row per halfspace $h_i \in Q$. The entry $(i,j)$ is $1$ if $p_j$ is in the halfspace $h_i$ and it is $0$ otherwise. Such matrices are known to have hereditary discrepancy $O(n^{1/4})$~\cite{matousek:half}.
\end{itemize}

The running times of the algorithms varied exclusively with the matrix size and not the type of matrix, thus we only show one time column which is representative of all types of matrices. The results are shown in Table~\ref{tbl}.
\begin{center}
\begin{table}[h]

\centering
\begin{tabular}{|c|c|r|r|r|r|}
\hline
Algorithm & Matrix Size & Disc Uniform & Disc 2D Corner & Disc 2D Halfspace & Time (s)\\
\hline
\Minimize & $200 \times 200$ & 24 & 3 & 4 & $< 1$\\
\Sample & $200 \times 200$ & 42 & 17 & 19 & $<1$\\
\SampleMany & $200 \times 200$ & 28 & 6 & 6 & $<1$\\
\hline
\Minimize & $1000 \times 1000$ & 56 & 6 & 8 & 24\\
\Sample &  $1000 \times 1000$ & 108 & 40 & 47 & $<1$\\
\SampleMany &  $1000 \times 1000$ & 80 & 15 & 18 & 24\\
\hline
\Minimize & $4000 \times 4000$ & 122 & 8 & 11 & 1637\\
\Sample &  $4000 \times 4000$ & 238 & 84 & 99 & $<1$\\
\SampleMany &  $4000 \times 4000$ & 186 & 32 & 36 & 1637\\
\hline
\Minimize & $10000 \times 2000$ & 124 & 9 & 11 & 1479\\
\Sample &  $10000 \times 2000$ & 178 & 61 & 71 & $<1$\\
\SampleMany &  $10000 \times 2000$ & 142 & 25 & 29 & 1479\\
\hline
\end{tabular}
\caption{Results of experiments with our \Minimize algorithm. The Matrix Size column gives the size $m \times n$ of the input matrix. The Disc columns show $\disc(A,x) = \|Ax\|_\infty$ for the coloring $x$ found by the algorithm on the given type of matrix. Time is measured in seconds.}
\label{tbl}
\end{table}
\end{center}
The table clearly shows that \Minimize gives superior colorings for all types of matrices and all sizes. The tendency is particularly clear on the structured matrices \textbf{2D Corner} and \textbf{2D Halfspace} where the coloring found by \Minimize on $4000 \times 4000$ matrices is a factor roughly 10 smaller than a single round of random sampling (\Sample) and a factor roughly 4 better than random sampling for as long time as \Minimize runs (\SampleMany).

\section{Deferred Proofs}
\label{sec:boring}
In this section, we give the deferred proof from
Section~\ref{sec:prelim} and also the reduction from $n > m$ to $m
\geq n$ case of discrepancy minimization.

\begin{proof}[Proof of Claim~\ref{claim:ipprojnormal}]
 Let $v_1,\dots,v_\ell$ denote the rows of $V$ and let
$u_1,\dots,u_{n-\ell}$ be an orthonormal basis for the orthogonal
complement of $v_1,\dots,v_\ell$. By the rotational invariance of the gaussian distribution, we have
 that $g$ is distributed as $\sum_{i=1}^\ell X_i v_i +
 \sum_{j=1}^{n-\ell} Y_j u_j$ where all $X_i$ and $Y_i$ are
 independent $\Norm(0,1)$ random variables. Write $a = \sum_{i=1}^\ell
 \alpha_i v_i + \sum_{j=1}^{n-\ell} \beta_j u_j$. Then
 \begin{eqnarray*}
   \langle a, (I-V^TV)g \rangle &=& \sum_{i=1}^\ell \alpha_i 
                                    v_i^T (I-V^TV)g +
                                    \sum_{j=1}^{n-\ell} \beta_j
                                    u_j^T(I-V^TV)g \\
   &=& \sum_{i=1}^\ell \alpha_i 
                                    (v_i^T - v_i^T)g +
                                    \sum_{j=1}^{n-\ell} \beta_j
       u_j^Tg \\
   &=& \sum_{j=1}^{n-\ell} \beta_j
       u_j^T\left(\sum_{i=1}^\ell X_i v_i + \sum_{h=1}^{n-\ell} Y_h
       u_h\right) \\
                                &=& \sum_{j=1}^{n-\ell} \beta_j Y_j  \\
                                &\sim& \Norm\left(0,\sum_{j=1}^{n-\ell} \beta_j^2\right) \\
   &\sim& \Norm\left(0,\|a(I-V^TV)\|^2\right).
 \end{eqnarray*}
\end{proof}

\paragraph{Reducing to $m \geq n$.}
Assume the input matrix $A$ has $m < n$. Initialize an empty matrix
$V$ ($0 \times n$) and add all rows of $A$ to $V$ using Orthogonalize
(Algorithm~\ref{alg:orth}). Also initialize a coloring $x = 0 \in
\R^n$. While the number of rows in $V$ is less
than $n$, there is a non-zero vector $g$ in the orthogonal complement
of the rows of $V$. Such a vector $g$ can e.g. be found by sampling
each coordinate as an independent $\Norm(0,1)$ distributed random
variable and then updating $g \gets g(I-V^TV)$. Given such a $g$,
observe that $\langle a_i, g \rangle = 0$ for all rows $a_i$ of
$A$. Now compute the largest coefficient $\eps$ such that $\|x + \eps
g\|_\infty = 1$, i.e. $\eps$ is the largest scaling such that at least
one coordinate of $x + g$ becomes $1$ in absolute value. Update $x$ by
adding $\eps g$ to it. For every entry $x_i$ that becomes $1$ in
absolute value, add $e_i$ to $V$ using
Algorithm~\ref{alg:orth}. Repeat until there are $n$ rows in
$V$.

We observe that when the above process terminates, we still have $Ax =
0$. Moreover, we must have added $n - m$ vectors $e_i$ to $V$ using
Algorithm~\ref{alg:orth}. Hence there are at most $n-(n-m) \leq m$
coordinates left that are not in $\{-1,1\}$. We may now use
PartialColoring starting from this vector $x$.

For the running time, first notice that adding all rows of $A$ to $V$
takes time $O(mn^2)$ since a call to Algorithm~\ref{alg:orth} takes
$O(n^2)$ time. Next, observe that each iteration adds at least one $e_i$
to $V$. Thus we have at most $n$ iterations. In any iteration,
computing $g$ takes $O(n^2)$ time. Computing the coefficient $\eps$
takes $O(mn)$ time and adding any $e_i$ to $V$ takes $O(n^2)$ time. Thus the total running time is $O(mn^2 +
n^3)$.

\bibliographystyle{abbrv}
\bibliography{bibliography}

\begin{thebibliography}{10}

\bibitem{alexander:half}
R.~Alexander.
\newblock Geometric methods in the study of irregularities of distribution.
\newblock {\em Combinatorica}, 10(2):115--136, 1990.

\bibitem{fastdisc}
R.~Alweiss, Y.~P. Liu, and M.~Sawhney.
\newblock Discrepancy minimization via a self-balancing walk.
\newblock In {\em {STOC} '21: 53rd Annual {ACM} {SIGACT} Symposium on Theory of
  Computing, 2021}, pages 14--20. {ACM}, 2021.

\bibitem{Banas:gauss}
W.~Banaszczyk.
\newblock Balancing vectors and gaussian measures of n-dimensional convex
  bodies.
\newblock {\em Random Structures \& Algorithms}, 12:351--360, July 1998.

\bibitem{bansal}
N.~Bansal.
\newblock Constructive algorithms for discrepancy minimization.
\newblock In {\em Proc. 51th Annual {IEEE} Symposium on Foundations of Computer
  Science (FOCS'10)}, pages 3--10, 2010.

\bibitem{bansalKomlos}
N.~Bansal, D.~Dadush, and S.~Garg.
\newblock An algorithm for koml{\'{o}}s conjecture matching banaszczyk's bound.
\newblock In {\em Proc. 57th {IEEE} Annual Symposium on Foundations of Computer
  Science (FOCS'16)}, pages 788--799, 2016.

\bibitem{beckfiala}
J.~Beck and T.~Fiala.
\newblock Integer-making theorems.
\newblock {\em Discrete Applied Mathematics}, 3:1--8, February 1981.

\bibitem{Charikar}
M.~Charikar, A.~Newman, and A.~Nikolov.
\newblock Tight hardness results for minimizing discrepancy.
\newblock In {\em Proc. 22nd Annual ACM-SIAM Symposium on Discrete Algorithms},
  SODA '11, pages 1607--1614, 2011.

\bibitem{chazelle:discrepancy}
B.~Chazelle.
\newblock {\em The Discrepancy Method: Randomness and Complexity}.
\newblock Cambridge University Press, 2000.

\bibitem{chazelleLvov}
B.~Chazelle and A.~Lvov.
\newblock A trace bound for the hereditary discrepancy.
\newblock In {\em Proc. 16th Annual Symposium on Computational Geometry}, SCG
  '00, pages 64--69, 2000.

\bibitem{balancing}
D.~Dadush, A.~Nikolov, K.~Talwar, and N.~Tomczak-Jaegermann.
\newblock Balancing vectors in any norm.
\newblock In {\em 2018 IEEE 59th Annual Symposium on Foundations of Computer
  Science (FOCS)}, pages 1--10, 2018.

\bibitem{Eldan2018EfficientAF}
R.~Eldan and M.~Singh.
\newblock Efficient algorithms for discrepancy minimization in convex sets.
\newblock {\em Random Structures \& Algorithms}, 53:289 -- 307, 2018.

\bibitem{sdp}
H.~Jiang, T.~Kathuria, Y.~T. Lee, S.~Padmanabhan, and Z.~Song.
\newblock A faster interior point method for semidefinite programming.
\newblock In {\em 61st {IEEE} Annual Symposium on Foundations of Computer
  Science, {FOCS} 2020}, pages 910--918. {IEEE}, 2020.

\bibitem{larsenDisc}
K.~G. Larsen.
\newblock On range searching in the group model and combinatorial discrepancy.
\newblock {\em {SIAM} Journal on Computing}, 43(2):673--686, 2014.

\bibitem{herel2}
K.~G. Larsen.
\newblock Constructive discrepancy minimization with hereditary {L2}
  guarantees.
\newblock In {\em 36th International Symposium on Theoretical Aspects of
  Computer Science, {STACS} 2019}, volume 126 of {\em LIPIcs}, pages
  48:1--48:13. Schloss Dagstuhl - Leibniz-Zentrum f{\"{u}}r Informatik, 2019.

\bibitem{lovasz}
L.~Lov{\'{a}}sz, J.~Spencer, and K.~Vesztergombi.
\newblock Discrepancy of set-systems and matrices.
\newblock {\em European Journal of Combinatorics}, 7(2):151 -- 160, 1986.

\bibitem{lovett}
S.~Lovett and R.~Meka.
\newblock Constructive discrepancy minimization by walking on the edges.
\newblock {\em {SIAM} Journal on Computing}, 44(5):1573--1582, 2015.

\bibitem{matousek:half}
J.~Matou{\v s}ek.
\newblock Tight upper bounds for the discrepancy of half-spaces.
\newblock {\em Discrete and Computational Geometry}, 13:593--601, 1995.

\bibitem{matousek1999geometric}
J.~Matousek.
\newblock {\em Geometric Discrepancy: An Illustrated Guide}.
\newblock Algorithms and Combinatorics. Springer Berlin Heidelberg, 1999.

\bibitem{matousekGamma}
J.~Matou{\v{s}}ek and A.~Nikolov.
\newblock {Combinatorial Discrepancy for Boxes via the gamma{\_}2 Norm}.
\newblock In {\em 31st International Symposium on Computational Geometry (SoCG
  2015)}, volume~34, pages 1--15, 2015.

\bibitem{factNorms}
J.~Matou{\v{s}}ek, A.~Nikolov, and K.~Talwar.
\newblock Factorization norms and hereditary discrepancy.
\newblock {\em CoRR}, abs/1408.1376, 2014.

\bibitem{gamma2}
J.~Matoušek, A.~Nikolov, and K.~Talwar.
\newblock {Factorization Norms and Hereditary Discrepancy}.
\newblock {\em International Mathematics Research Notices}, 2020(3):751--780,
  03 2018.

\bibitem{rothvoss}
T.~Rothvoss.
\newblock Constructive discrepancy minimization for convex sets.
\newblock {\em SIAM Journal on Computing}, 46(1):224--234, 2017.

\bibitem{azuma}
O.~Shamir.
\newblock A variant of azuma's inequality for martingales with subgaussian
  tails.
\newblock {\em CoRR}, abs/1110.2392, 2011.

\bibitem{spencer}
J.~Spencer.
\newblock Six standard deviations suffice.
\newblock {\em Trans. Amer. Math. Soc.}, 289:679--706, 1985.

\bibitem{srinivasan}
A.~Srinivasan.
\newblock Improving the discrepancy bound for sparse matrices: Better
  approximations for sparse lattice approximation problems.
\newblock In {\em Proceedings of the Eighth Annual {ACM-SIAM} Symposium on
  Discrete Algorithms, 1997}, pages 692--701. {ACM/SIAM}, 1997.

\end{thebibliography}

\end{document}